\documentclass[10pt,a4paper]{article}
\usepackage{amsmath,amssymb,amsfonts,amsthm}
\usepackage{graphicx}
\usepackage{t1enc}
\usepackage[english]{babel}
\usepackage[latin1]{inputenc}
\usepackage{color}
\usepackage{verbatim}
\usepackage{hyperref}
\usepackage{color}
\usepackage{booktabs}
\usepackage{appendix}

\newtheorem{theorem}{Theorem}[section]

\theoremstyle{definition}

\usepackage{enumerate}

\newcommand{\be}{\begin{equation}}

\newcommand{\ee}{\end{equation}}

\newcommand{\Rt}{\mathbb{R}^3}

\newcommand{\RA}{\mathcal{R}_A}
\newcommand{\RC}{\mathcal{R}_C}

\newcommand{\R}{\mathcal{R}}

\usepackage{vmargin}

\setpapersize{A4}
\setmargins{3.5cm}       
{1.8cm}                        
{14.5cm}                      
{23.12cm}                    
{10pt}                           
{1cm}                           
{0pt}                             
{2cm}                           

\title{Comments on Penrose inequality with angular momentum for outermost apparent horizons.}
\author{Pablo Anglada \footnote{panglada@famaf.unc.edu.ar}\\
  Facultad de Matem\'atica, Astronom\'{i}a, F\'{i}sica, y Computaci\'on  \\
     Universidad Nacional de C\'ordoba, \\
Instituto de F\'{i}sica Enrique Gaviola, IFEG, CONICET,\\
  Ciudad Universitaria (5000) C\'ordoba, Argentina.}

\begin{document}
\maketitle
\begin{abstract}
In a recent work
we have proved a weaker version of the Penrose inequality with angular momentum, in axially symmetric space-times, for a compact and connected minimal surface. In this previous work we use the monotonicity of Geroch energy on 2-surfaces along the inverse mean curvature flow and we obtain a lower bound for the ADM mass in terms of the area, the angular momentun and a particular measure of size of the minimal surface. 
In the present work, using similar techniques and the same measure of size, we extend and improve the previous result for a compact and connected outermost apparent horizon. For this case we use the monotonicity of Hawking energy, instead of Geroch energy, along the inverse mean curvature flow, and assume different conditions on the extrinsic curvature. This type of relations constitutes an important test to evaluate the cosmic censorship conjecture.

\end{abstract}

\section{Introduction}
 
After formulating the \textit{cosmic censorship conjecture},
Penrose proposed \cite{Penrose1973} that, when considering
collapsing matter,  if the conjecture is valid, the mass $m$ and the area $A$ of the resulting black hole must satisfy the relation:

\be\label{penroseineq}
m\geq \sqrt{\frac{A}{16\pi}}
\ee

Given that the validity of this relation gives one of the most important tests to evaluate the cosmic censorship conjecture, there are plenty of works that study this relation, and one can find very exhaustive review articles \cite{Mars2009, Bray2003} as well as general approaches to the problem  \cite{Frauendiener:01, BrayKhuri:11, BrayKhuri:09}. Moreover, since the Penrose heuristic argument one can strengthen the to include charge and angular momentum and charge (see \cite{Dain:2014Geo} \cite{Dain:2012} \cite{Mars2009} for more details), this topic has become an active area of research.
Good progress has been made in considering the case of a charged black hole without angular momentum \cite{Weinstein:2005,Disconzi:2012, KhuriYamada:2013, Khuri:2013, Khuri:2014, Khuri:2015}, but regarding the case with angular momentum, there are only a few results exploring the relation between the angular momentum, the size and the mass of a compact object \cite{Anglada:2016dbu, Anglada:2017ryp, Jaracz-Khuri18}. For further details and references regarding geometrical inequalities bounding angular momentum see the review article \cite{Dain:2017jkj}. In this work we present an extension of our previous work \cite{Anglada:2017ryp} for compact and connected general horizons.

Take an axially symmetric initial data $M$ connecting the black hole region with spatial infinity , such that the collapse has already occurred, and calculate the mass $m$, the area $A$, and the angular momentum $J$ of the black hole. Then, from the Penrose heuristic argument for rotating black holes, see for example \cite{Anglada:2017ryp}, we expect that:
\be \label{PenJ}
m^2 \geq \frac{A }{16 \pi} + \frac{4 \pi J^2}{A}
\ee
Note that this version of the Penrose inequality admits a rigidity case which states that the equality can only occur for the Kerr black hole.

In \cite{Anglada:2017ryp} we studied this problem in the particular case that the apparent horizon is a compact and connected minimal surface.
We use the monotonicity properties of the Geroch energy \cite{Geroch1973} along the IMCF and proved the following version of \eqref{PenJ}:
\be\label{PenJms}
m_{ADM}^2\geq \frac{A }{16 \pi} +  \frac{ J^2}{ 2\R^2}
\ee
where $m_{ADM}$ is the ADM mass, \cite{ADM}, and $\R$ is a specific measure of size defined in terms of the norm of the axial Killing vector. This measure has reasonably nice properties, see \cite{Anglada:2016dbu, Anglada:2017ryp}, and under certain conditions can be related to usual measures of size.
In this previous work, in order to have a non-negativity scalar curvature we need to assume some special conditions for the extrinsic curvature, for example that the initial data is maximal.
In this work we use Hawking energy, instead of Geroch energy, similar techniques and the same measure of size to extend and improve \eqref{PenJms} for a compact and connected general horizon. 

Regarding the mentioned papers, \cite{Anglada:2016dbu, Anglada:2017ryp, Jaracz-Khuri18}, that explore the relation between the angular momentum, the size and the mass of a compact object it is important to remark the following issue. Although the technique used to relate the angular momentum with the surface integral of the extrinsic curvature are very similar, and the particular measure of the axial radii is the same, in \cite{Anglada:2016dbu}, \cite{ Anglada:2017ryp}, and in the present work, one need to impose some particular conditions over the extrinsic curvature in order to assure that the energy is non-decreasing along the flow, and also one need to assume there exist a smooth solution of the IMCF for the initial data. In \cite{Jaracz-Khuri18} the author present a different approach to this problem, they use an embellished version of the Jang equation and they study a Jang/IMCF system of equations. This allows the authors to obtain a result very similar to the one presented on this work but assuming only a few general conditions, but they need assume that there exist a smooth solution of a more complex system of equations, where the IMCF equation and a particular version of Jang equation are coupled.

\section{Background}
\label{Background}

We consider an \textit{asymptotically flat} and \textit{axially symmetric} initial data set $(M, \partial M,\bar g, K; \mu, j^i)$ with boundary $\partial M$, where $M$ is a 3-manifold with positive definite metric $\bar g$ and extrinsic curvature $K$, $\partial M$ is a connected and compact 2-surface, $\mu$ is the energy density and $j^i$ is the matter current density. This set must satisfy the constraint equations 
\begin{align}
 \label{const1}
   \bar D_j   K^{i j} -  \bar D^i   k= -8\pi j^i,\\
 \label{const2}
   \bar R -  K_{i j}   K^{i j}+  k^2=16\pi \mu,
\end{align}
where ${\bar D}$ and $\bar R$ are the Levi-Civita connection and the curvature scalar associated with $  \bar g$, and $k=\mbox{tr}_{\bar g}K$. 
We assume the matter fields satisfy the \textit{Dominant Energy Condition} (DEC), $\mu \geq |j|$,  and that $\partial M$ is an outermost future apparent horizon and there are no other trapped surfaces on $M$. With these assumptions $M$ is an \textit{exterior region} and has the topology $\Rt$ minus a ball \cite{Huisken2001}.

Assume there exists a smooth inverse mean curvature flow (IMCF) of surfaces $S_t$ starting from $S_{0}=\partial M$ and having spherical topology, and take $\nu$ to be unit normal vector of $S_t$ (see equation \eqref{eqIMCF} on Appendix A). Then one can write the metric $\bar g$ in the form:
\be\label{barg}
ds^2_{\bar g}=\frac{dt^2}{H^2} + g_{ij}dx^i dx^j
\ee
where $H$ and $g_{ij}$ are the mean curvature and the induced metric of $S_t$ respectively, and $(x^1, x^2)$ are and the induced coordinates.
See \cite{Huiskenevol} \cite{Szabados04} for a review of the basic properties of the IMCF.
In this context the extrinsic curvature can be decomposed \cite{Malec:2002ki}:
\be\label{K}
K_{ij}=z\nu_i \nu_j + \nu_i s_j + s_i \nu_j  + g^k_i g^l_j \chi_{lk} + \frac{q}{2}g_{ij}
\ee
where q is the trace with respect to $g_{ij}$ of $K$, $q=K_{ij}g^{ij}$ and 
\be\label{Kparts}
z=K_{ij}\nu^i \nu^j \qquad s_i=g^j_i K_{jl} \nu^l \qquad \chi_{ij}=g^l_i g^n_j K_{ln} - \frac{q}{2} g_{ij}
\ee
then the trace of the extrinsic curvature takes the form $k=\mbox{tr}(K)=z+q$ and its norm is
\be\label{KK}
K_{ij}K^{ij}=z^2+2s_i s^i + \chi_{ij} \chi^{ij} + \frac{q^2}{2}.
\ee
and then from equation \eqref{const2} we have that the scalar curvature of the initial data can be written in the following way:
\begin{equation}
\label{barR}
\bar  R= 16\pi \mu + 2s_i s^i + \chi_{ij} \chi^{ij} -  \frac{q}{2}(q + 4z).
\end{equation}

Let $\vartheta^+$ and $ \vartheta^-$ be the expansions of the outgoing null geodesics orthogonal to $S_t$, future directed and past directed respectively, then because $\partial M$ is a future apparent horizon we have $\vartheta^+|_{\partial M}=0$ and $\vartheta^-|_{\partial M}\geq0$.
From the previous decomposition $ \vartheta^+|_{S_t}=H+q$ and  $\vartheta^-|_{S_t}=H-q$ and then
if $M$ has no other trapped surface than $\partial M$, the expansions satisfy satisfy $\vartheta^+|_{S_t} \geq 0$, $\vartheta^-|_{S_t}\geq 0$ $\forall t >0$, and then 
\be\label{NTS}
(\vartheta^+\vartheta^-)|_{S_t}=H^2-q^2>0 \quad \forall t > 0,
\ee

Following \cite{Malec:2002ki} we are going to use a functional proposed by Hawking \cite{Hawking:1968qt},the Hawking energy of a surface $E_H(S)$:
\be \label{Hmass-gral}
E_H(S):=\frac{A^{1/2}}{(16\pi)^{3/2}}\left(16\pi-\int_{S}\vartheta^+\vartheta^-dS\right)
\ee
where $A$ is the area of $S$. 
This energy, under certain conditions (see \cite{Malec:2002ki}) is monotonic under a smooth inverse mean curvature flow, and has the interesting properties that it tends to the ADM mass of $M$ at infinity and for a future apparent horizon is equal to $\sqrt{\frac{A}{16 \pi}}$. 
From \cite{Malec:2002ki} we have that the derivate of Hawking energy along the IMCF can be written in the following way (see appendix \ref{sec:IMCF} for more details and definitions of used quantities): 
\begin{multline}
\label{dtEgralf}
\frac{d}{dt}E_H=\frac{A_t^{1/2}}{(16\pi)^{3/2}}\left[ \int_{S_t} \left(16 \pi ( \mu +\frac{q}{H}\nu^i j_i)  \right) + \int_{S_t} 2 \frac{q}{H} g^{ij} \bar \nabla_i s_j dS
\right.\\
\left.              +\int_{S_t} \left(\chi_{ij} \chi^{ij} -  2\frac{q}{H}\chi^{ij} t_{ij}  + t_{ij}   t^{ij}\right)dS \right.\\
\left.              +\int_{S_t} 2\left(s_i s^i - 2 \frac{q}{H} s^j \frac{ \bar \nabla_j H}{H} + \frac{ g^{ij} \bar \nabla_i H \bar \nabla_j H}{H^2} \right)dS \right]
\end{multline}

\section{Main result}

Now from \cite{Anglada:2016dbu} we know that when considering the IMCF in axially symmetric initial data, 
the IMCF equation preserves axial symmetry. Then from now on, when we discuss the IMCF flow, we always consider it consisting of axially symmetric surfaces $S_t$. Then for each surface of the flow we can define orthogonal coordinates $\theta, \varphi$ such that $\eta^i=\frac{\partial}{\partial \varphi}^i$. One can always choose this for axially symmetric 2-surfaces that are diffeomorphic to $S^2$, see for example
\cite{Dain:2011pi}. Hence we have:
\be \label{gSt}
  ds^2_g=\Psi^4 d\theta^2+ \eta d\varphi^2
\ee
where $\eta=g_{ij}\eta^i \eta^j$ is the square norm of the axial Killing vector. 

The physical and geometrical quantities we are interested in are the ADM mass $m_{ADM}$ and the Komar angular momentum $J(S_t)$: 
\be\label{defangmom}
 J(S_t)=\frac{1}{8\pi}\int_{S_t} K_{ij} \eta^i \nu^jdS,
\ee
where we use that $\bar g_{ij} \nu^i \eta^j=0$.

To measure the size of the surface $S_t$ we will use the areal and circumferential radii of a surface $S_t$ in $M$:
\be\label{size}
\RA(S_t):=\sqrt{\frac{A_t}{4\pi}},\qquad \RC(S_t):=\frac{\mathcal C (S_t)}{2\pi}
\ee
where $A_t$ is the area of $S_t$ and $\mathcal C(S_t)$ is the length of the greatest axially symmetric circle in $S_t$.
It is also useful to consider the following size measure studied in \cite{Anglada:2016dbu, Anglada:2017ryp}:
\be\label{R}
\frac{1}{\R(S_t)^2} := A_t^{1/2} \int_{t}^\infty \frac{A_{t'}^{1/2}}{\int_{S_{t'}}  \eta dS }dt'
\ee

This measure of size of a surface $S_t$, based on the behavior of the norm of the Killing vector along the IMCF from $S_t$ to infinity, is positive and well defined provided the flow remains smooth. Moreover, as shown in \cite{Anglada:2017ryp}, in some cases, $\R$ can be related to $R_A$ and  $R_C$. In particular assuming that the IMCF is convex we have:
\be
\R^2(S_t) \leq \frac{5}{2}\frac{\int_{S_t} \eta dS}{A_t} \leq \frac{5}{2} \RC^2(S_t),
\ee

Using the previous tools and this definition of size, and assuming the same conditions of the main theorem in \cite{Malec:2002ki}, we prove the following theorem.

\begin{theorem}
\label{theo1}
Let $(M,\partial M, \bar g, K)$ be a vacuum, asymptotically flat, and axially symmetric initial data, such that $\partial M$ is a compact and connected outermost apparent horizon and there are no other trapped surfaces on $M$. Assume there exists a smooth IMCF of surfaces $S_t$ starting from $\partial M$ and having spherical topology, then if the initial data satisfies either:
\begin{enumerate}[a)]
        \item \label{a} $\frac{q}{H}$ is constant for each surface $S_t$ , or 
        \item \label{b} $g^{ij} \bar \nabla_i s_j=0$,
\end{enumerate}
then:
\be\label{teo1}
m_{ADM}^2 \geq \frac{A}{16 \pi} + \frac{ J^2}{\R^2}
\ee
where $J$ and $A$ are the angular momentum and the area of $\partial M$ respectively, and $\R=\R(\partial M)$ is defined by \eqref{R}.
\end{theorem}

\begin{proof}

From the fact that the surfaces of the flow are axially symmetric we have that mean curvature $H$ does not depend on the coordinate $\varphi$:
\be
s^j \bar \nabla_j H = s^\theta \bar \nabla_\theta H \qquad  g^{ij} \bar \nabla_i H \bar \nabla_j H= g^{\theta \theta} \bar \nabla_\theta H \bar \nabla_\theta H
\ee
hence in this case:
\begin{multline}
\frac{d}{dt}E_H=\frac{A_t^{1/2}}{(16\pi)^{3/2}}\left[ \int_{S_t} 2 \frac{q}{H} g^{ij} \bar \nabla_i s_j dS 
+\int_{S_t} \left(\chi_{ij} \chi^{ij} -  2\frac{q}{H}\chi^{ij} t_{ij}  + t_{ij}   t^{ij}\right)dS
\right.\\ 
\left.    +\int_{S_t} 2\left(s_\theta s^\theta - 2 \frac{q}{H^2} s^\theta  \bar \nabla_\theta H + \frac{1}{H^2} g^{\theta\theta} \bar \nabla_\theta H \bar \nabla_\theta H  \right)dS \right] + 2\frac{A_t^{1/2}}{(16\pi)^{3/2}}\int_{S_t}  s_\varphi s^\varphi dS
\end{multline}
where we have used that the initial data is vacuum.

In order to include the angular momentum into the inequality we know from \cite{Anglada:2017ryp} that we can relate the angular momentum of any surface $S_t$ to the surface integral of the norm of $s_i$ and the norm of $\eta^i$.  In this work we need to improve the previous calculation in order to relate the angular momentum, not with the norm of $s_i$, but only with the component of $s_i$ along the axial Killing vector. First note that $K_{ij} \eta^i \nu^j=s_i \eta^i=s_\varphi$, then using the Cauchy-Schwarz inequality in the definition of $J_t:=J(S_t)$:  
\be\label{angmom}
\begin{split}
 J_t^2&=\frac{1}{(8\pi)^2} \left( \int_{S_t} s_i \eta^i dS\right)^2 = \frac{1}{(8\pi)^2} \left( \int_{S_t} s_\varphi dS\right)^2 \\
 &\leq \frac{1}{(8\pi)^2} \left( \int_{S_t} |s_\varphi| \frac{\sqrt{ \eta}}{\sqrt{\eta}} dS\right)^2 \leq \frac{1}{(8\pi)^2} \int_{S_t} \frac{s_\varphi^2}{\eta}dS  \int_{S_t}\eta dS \\
 &= \frac{1}{(8\pi)^2} \int_{S_t} s_\varphi s^\varphi dS  \int_{S_t}\eta dS
\end{split}
\ee
where in the fourth step we have used the H\"older inequality with $p_1=p_2=2$. Then we have:
\be
\int_{S_t} s_\varphi s^\varphi dS \geq \frac{(8\pi)^2 J_t^2}{\int_{S_t}\eta dS}
\ee 
and hence we could include explicitly the angular momentum on the derivate of Hawking energy:

\begin{multline}\label{dtEvsJgral}
\frac{d}{dt}E_H \geq\frac{A_t^{1/2}}{(16\pi)^{3/2}}\left[ \int_{S_t} 2 \frac{q}{H} g^{ij} \bar \nabla_i s_j dS 
+\int_{S_t} \left(\chi_{ij} \chi^{ij} -  2\frac{q}{H}\chi^{ij} t_{ij}  + t_{ij}   t^{ij}\right)dS
\right.\\ 
\left.    +\int_{S_t} 2\left(s_\theta s^\theta - 2 \frac{q}{H^2} s^\theta  \bar \nabla_\theta H + \frac{1}{H^2} g^{\theta\theta} \bar \nabla_\theta H \bar \nabla_\theta H \right)dS \right] +2\sqrt{\pi} J_t^2 \frac{A_t^{1/2}}{\int_{S_t}\eta dS}
\end{multline}

Now first note that from the hypothesis that there are no other trapped surfaces in $M$ than $\partial M$ we have that $H^2>q^2$ and thus $|\frac{H}{q}| \geq 1$, hence the integrands on the second  and third terms $$\left(\chi_{ij} \chi^{ij} -  2\frac{q}{H}\chi^{ij} t_{ij}  + t_{ij}   t^{ij}\right), \quad \left( s_\theta s^\theta - 2 \frac{q}{H^2} s^\theta  \bar \nabla_\theta H + \frac{1}{H^2} g^{\theta\theta} \bar \nabla_\theta H \bar \nabla_\theta H\right) $$ are positive quadratic forms, and thus the second and third integrals in \eqref{dtEvsJgral} are positive. Then note that if we assume condition \ref{b} the integrand on the first term in \eqref{dtEvsJgral} is equal to zero, and if we assume condition \ref{a} then by partial integration on $S_t$ the first term in \eqref{dtEvsJgral} also vanishes.

Thus assuming the hypothesis of the theorem and either of the conditions \ref{a} or \ref{b}, the first term in \eqref{dtEvsJgral} vanishes and the second and third terms are positive, hence:
\be\label{dtEvsJ}
\frac{d}{dt}E_H\geq \sqrt{4\pi} J^2 \frac{A_t^{1/2}}{\int_{S_t}\eta dS}
\ee
where we have used that $M$ is a vacuum exterior region, thus $J_t=J$.

From these arguments we have that $E_H$ is monotonically increasing along the flow, hence $E_H(S_t)\geq E_H(\partial M) \quad \forall t\geq 0 $, then because $\partial M$ is an apparent horizon $H^2=q^2$ we have $E_H(\partial M)=\sqrt{\frac{A}{16\pi}}$ and thus:
\be
\label{Penrose}
E_H(S_t)\geq \sqrt{\frac{A}{16\pi}} \quad \forall t\geq 0 
\ee

Now we calculate the derivate along the flow of the functional $E_H^2$ and use equation \eqref{dtEvsJ} to obtain a lower bound for it in terms of $J$:
\be
\frac{d}{dt}E_H^2 =2E_H(S_t) \frac{d}{dt}E_H  \geq   2E_H(S_t) \sqrt{4\pi} J^2 \frac{A_t^{1/2}}{\int_{S_t}\eta dS}
\ee
then using equation \eqref{Penrose} we have:
\be\label{dtE2vsJ}
\frac{d}{dt}E_H^2  \geq    J^2 \sqrt{A} \frac{A_t^{1/2}}{\int_{S_t}\eta dS}
\ee
Now integrating this expression along the flow from $\partial M$ to infinity and using the relation between Hawking energy and the ADM mass we obtain:
\be\label{evol4}
m_{ADM}^2\geq\lim_{t \to \infty} E_H^2(S_t) \geq E_H^2(S_0) + J^2 \sqrt{A} \int_0^\infty \frac{A_t^{1/2}}{\int_{S_t}  \eta dS}dt.
\ee 

Finally we use the fact that $E_H^2(S_0)=\frac{A}{16\pi}$, and we write this expression in terms of $\R$ and obtain \eqref{teo1}.

\end{proof}

Inequality \eqref{teo1} is also valid for non-vacuum initial data, provided that the matter fields satisfy the DEC and that $j_i\eta^i=0$ everywhere in $M$.
Assuming the DEC we assure that Hawking energy remains monotonic for non-vacuum initial data. Condition $j_i\eta^i=0$ assures that the angular momentum is preserved along the flow $J(S_t)=J$, and that there is no contribution to $J$ coming from the matter fields $J=J(\partial M)$.

In case we have a non-zero contribution of the matter fields to the angular momentum, $j_i\eta^i \neq0$, we obtain an extension, for objects that contain a general horizon, of the results for ordinary objects presented in \cite{Anglada:2016dbu}. We assume that the matters fields satisfy the DEC and that both the matter density and the matter current have compact support.
In this case the angular momentum of a surface $S_t$ is
\be
 J(S_t)=\frac{1}{8\pi}\int_{S_t} K_{ij} \eta^i \nu^jdS=  J(\partial M) - \int_{V(S_t)} j_{i} \eta^i dv,
\ee
where $V(S_t)$ is the region enclosed between $\partial M$ and $S_t$. Thus the conservation of the angular momentum along the flow  is only satisfied when the surfaces $S_t$ are outside the compact support of the matter fields. Then the measures of size involved in the rotational contribution to the energy are not measures of size of the apparent horizon, but measures of size of the first surface of the flow $S_T$ that enclosed the object.

Then for a non-vacuum initial data with $j_i\eta^i \neq0$, assuming the same conditions of theorem \eqref{theo1}, we obtain the following result.

\begin{theorem}

Let $(M,\partial M, \bar g, K; \mu, j^i)$ be an initial data satisfying the same conditions of theorem  \ref{theo1}. Assume the matter fields satisfy the dominant energy condition and have compact support, and let $T$ such that for all $t\geq T$ the matter density and the matter current have compact support inside $S_t$, then:

\be\label{teo2}
m_{ADM} \geq  m_T + \frac{\R_A}{2} + \frac{ J^2}{\R_A(T)\R^2(T)}
\ee
where $J$ is the total angular momentum of the data, $\R_A$ and $\R_A$(T) are the areal radii of $\partial M$ and $S_T$ respectively, $\R(T)=\R(S_T)$ is defined  by \eqref{R}, and 
\be\label{mTbh}
m_{T}:=\int_{\RA}^{\RA(T)}\int_{S_\xi} \left( \mu +\frac{q}{H}\nu^i j_i \right) dS d\xi 
\ee
where $\xi$ stands for the  areal radius coordinate.
\end{theorem}

\begin{proof}
From \eqref{dtEgralf} and the previous calculations we have:
\be\label{dtEvsJ2}
\frac{d}{dt}E_H\geq  \sqrt{\frac{A_t}{16\pi}} \int_{S_t} \left( \mu +\frac{q}{H}\nu^i j_i  \right) dS +\sqrt{4\pi} J_t^2 \frac{A_t^{1/2}}{\int_{S_t}\eta dS}
\ee
Note that because the matter fields satisfy the DEC, and $|\frac{q}{H}|\leq 1$ the fist term in \eqref{dtEvsJ2} is also positive.
Then, integrating this expression along the flow from $\partial M$ to infinity and using the relation between Hawking energy and the ADM mass we have:
\be
m_{ADM}\geq E_H(S_0) + \int_0^\infty \sqrt{\frac{A_t}{16\pi}} \int_{S_t} \left(  \mu +\frac{q}{H}\nu^i j_i\right) dS +\int_0^\infty \sqrt{4\pi} J_t^2 \frac{A_t^{1/2}}{\int_{S_t}\eta ds}dt 
\ee

Now because the matter fields have compact support inside $S_T$ the fist integral runs only from $0$ to $T$, then dividing the integral involving the angular momentum from $0$ to $T$ and from $T$ to infinity and using that $J_t=J \quad \forall t\geq T$ we obtain
\begin{multline}
m_{ADM} \geq E_H(S_0) + \int_0^T \sqrt{\frac{A_t}{16\pi}} \int_{S_t} \left( \mu +\frac{q}{H}\nu^i j_i  \right) dS dt \\ 
+ \sqrt{4\pi} J^2 \int_T^\infty \frac{A_t^{1/2}}{\int_{S_t}  \eta dS}dt + \sqrt{4\pi} \int_0^T J_t^2 \frac{A_t^{1/2}}{\int_{S_t}\eta dS} dt . 
\end{multline}
Finally disregarding the last term, using that $E_H(S_0)=\frac{\R_A}{2}$, and writing this expression in terms of $\R_A(S_T)$, $\R(S_T)$ and $m_T$ we obtain \eqref{teo2}.

\end{proof}

\textbf{Remarks} 

\vspace{0.5cm}
The notion of size we use, $\R$, albeit apparently artificial at first sight
have proved to be very useful to relate the angular momentum to the total mass in axially symmetric and asymptotically flat initial data \cite{Anglada:2017ryp}. It comes from the particular method we use to relate the angular momentum with the ADM
mass, and gives a good measure of how different the IMCF is from a spherical one.
These kind of measures based on the norm of
the Killing vector have been found to give an appropriate description of size of a
region when describing both regular objects and black holes with angular
momentum \cite{Anglada:2016dbu}, \cite{Reiris:2014tva}, \cite{Reiris:2013jaa},
\cite{Dain:2014}.

Assuming particular properties for the IMCF we can write \eqref{teo1} in terms of the usual measures of size.
The best situation is to have a spherical IMCF, thus $\R^2(\partial M)=\frac{A}{4\pi}$, in which case this proof implies the validity of \eqref{PenJ} in the vaccum case:
\be
\label{inqA}
m_{ADM}^2\geq \frac{A}{16 \pi} + \frac{4 \pi J^2}{A}.
\ee
In general we do not expect to have a spherical IMCF in the context we are considering.
For weaker conditions for the IMCF, for example assuming that the flow is convex, we obtain a weaker version of the Penrose inequality with angular momentum for vaccum initial data in terms of the axial radius of the apparent horizon $\R_C=\R_C(\partial M) $:
\be
 m_{ADM}^2\geq \frac{A}{16 \pi} + \frac{2}{5}\frac{ J^2}{\RC^2}.
\ee

\vspace{0.5cm}

If we consider a spherically symmetric initial data then the angular momentum is zero, and the presented result implies the validity of the usual Penrose inequality for apparent horizons in spherical symmetry that was already proved by Malec and O'Murchadha in \cite{Malec:1994}. It is important to mention that in \cite{Ben-Dov:2004} Ben-Dov presents an important counterexample to some formulation of the Penrose inequality in spherical symmetry, but these counterexamples do not contradict either the result of Malec and O'Murchadha or the present result. In \cite{Ben-Dov:2004} the author explicitly constructs an asymptotically flat initial data that contains an apparent horizon and violates the Penrose inequality, but this particular initial data contains past-trapped surfaces outside the apparent horizon, and thus it does not satisfy the hypothesis of theorem \ref{teo1}. From this results one can infer that, if we consider a future apparent horizon $\partial M$, the hypothesis that states that there are no other trapped surfaces in $M$ other than $\partial M$ is necessary to assure the validity of the Penrose inequality.

\vspace{0.5cm}

In the previous work \cite{Anglada:2017ryp} we consider $\partial M$ to be a minimal surface and we use the monotonicity property of Geroch energy. To assure that this monotonicity property is valid we need to have a non-negativity scalar curvature. In order to do this we need to assume some special conditions for the extrinsic curvature, one possible choice is to take $K$ such that for every surface of the flow $q|_{S_t}=0$. This condition, known as the polar gauge condition, together with the assumption that there are no other trapped surfaces in $M$ than $\partial M$, also assures that $\partial M$ is an outermost future apparent horizon. Note that if we assume the condition $q|_{S_t}=0$ then Geroch and Hawking energies are equal along the IMCF, and in this case the previous and present results are the same. 
The other possible condition, the most usual, is to consider that the initial data is maximal, $k=0$. With this condition we have a non-negative scalar curvature, and thus Geroch energy is monotonic, but in this case the a minimal surface is a future apparent horizon only if the extrinsic curvature also satisfies $q|_{\partial M}=z|_{\partial M}=0$.
In present work we study the case in which $\partial M$ is a general outermost future apparent horizon, in this case the Geroch energy of $\partial M$ is given by the area of the horizon plus a term that involves the surface integral of $q^2$. Thus Geroch energy is no longer useful to study the problem, we need to use Hawking energy and the maximal condition is not sufficient to assure its monotonicity. The main problem is that one need to control the sing of the second term in \eqref{dtEgralf} to assure that the derivate along the flow of Hawking energy is positive. One possible option, condition \ref{a}, is to control the behavior of $q$ along the IMCF, the other option, condition \ref{b}, is to control the $\nu, \eta$ components of the extrinsic curvature.
Although \ref{a} and \ref{b} are less general than the maximal condition, because both of them depend on a particular foliation of the initial data given by the IMCF, these conditions allows us to extend the previous result to the case in which $q|_{\partial M}\neq 0$. At the present time we do not have a clear geometrical or physical interpretation of these conditions and it will be interesting to study its meaning in detail.

\vspace{0.5cm}

Condition \ref{b} can be fulfilled by choosing a particular form for $s_i$. Note that this condition is a necessary condition to get the monotonicity of Hawking energy if we do not want to assume the very restricted condition \ref{a}. First if we take $s_i$ such that does not have any component on the $\theta$ direction, that is to say $s^i=s\eta^i$ condition \ref{b} can be written in the following way:
\be
\begin{split}
g^{ij} \bar \nabla_i s_j =& g^{ij} \eta_j \bar \nabla_i s + s g^{ij}  \bar \nabla_i \eta_j \\
                         =& \eta^i\bar \nabla_i s +  \frac{ s}{2} \left( g^{ij} \bar \nabla_i \eta_j -  g^{ij} \bar \nabla_j \eta_i \right) \\
                         =& \partial_\varphi s \\
\end{split}
\ee
where in the second step we use the Killing equation for $\eta^i$, $\bar \nabla_i \eta_j =- \bar \nabla_j \eta_i$. Then one of the possible choices to get condition \ref{b} is to assume that $s^i=s\eta^i$ and that $s$ does not depend on the coordinate $\varphi$.

Is important to note that if we assume that $s_i=0$, conditions \ref{a} and \ref{b} are not necessary to get the monotonicity of Hawking energy, but in this case we do not have angular momentum. Take an asymptotically flat and axially symmetric initial data that do not have any other trapped surface than $\partial M$, satisfy the DEC and have $s_i=0$. Then for this initial data the existence of a smooth solution of the IMCF is the only necessary condition one needs to prove the positivity of the derivative of Hawking energy along the flow, and thus the only necessary condition one needs to prove the Penrose inequality \eqref{penroseineq}. Then one can infer that the angular momentum generates difficulties in obtaining a foliation of $M$ for which it can be assured that Hawking energy is monotonically increasing.

\vspace{0.5cm}
In respect to the assumption of existence of a smooth solution of the IMCF, the conditions we assume to assure that Hawking energy is monotonic will probably not be fulfilled for a weak flow. Moreover the method we use to relate the angular momentum with the energy strongly depends on having a smooth IMCF.
In order to adapt the present formulation to a weak formulation of the flow one need to control the behavior of all the terms involved in equation on each discontinuity of the flow, and we do not know if it is possible to control the terms involving $q$. In this sense we think the method presented in the previous work \cite{Anglada:2017ryp} has better chances to be adapted to a weak formulation of the IMCF   

\section*{Acknowledgments}
I thank Edward Malec and Marc Mars for enlightening discussions.
I would also like to thank Maria E. Gabach-Clement and Omar E. Ortiz for their encouraging support, and I thank the anonymous referees for their helpful and detailed reviews.
This work was partially supported by grants from CONICET and
SECyT, UNC.

\appendix

\section{Monotonicity of Hawking energy along the IMCF}\label{sec:IMCF}

Since it is relevant for the prove of our main theorem, in this section we will review the proof of the monotonicity property of Hawking energy obtained by Malec, Mars and Simon in \cite{Malec:2002ki}.

Let $(M, \partial M,\bar g, K; \mu, j^i)$ be a asymptotically flat and axially symmetric initial data with boundary. Assume there exists a smooth inverse mean curvature flow (IMCF) of surfaces $S_t$ starting from $S_{0}=\partial M$ and having spherical topology. Then we have a smooth family of hypersurfaces  $S_t:=x(S,t)$ on $M$, with $x:S\times[0,\tau]\to M$ satisfying the evolution equation
\be \label{eqIMCF}
\frac{\partial x}{\partial t}=\frac{\nu}{H}
\ee
where $t\in[0,\tau]$, $H>0$ is the mean curvature of the 2-surface $S_t$ at $x$ and $\nu$ is the outward unit normal to $S_t$.
Let  $\nabla_i$ be the covariant derivative, $h_{ij}$ the second fundamental form and $dS$ the area element of $S_t$.  Then one can derive the evolution equations, see \cite{Huiskenevol}, \cite{Szabados04}:
\be\label{evolg}
\frac{\partial}{\partial t}g_{ij}=\frac{2}{H}h_{ij}
\ee
\be\label{evolarea}
\frac{\partial}{\partial t}(dS)= dS
\ee
\be\label{evolH}
\frac{\partial}{\partial t}H=-\Delta(H^{-1})-H^{-1}(h_{ij}h^{ij}+\bar R_{ij}\nu^i\nu^j).
\ee

Using the decomposition of $\bar g$ and $K$ presented in section \ref{Background} we now calculate the derivate of Hawking:
\be\label{dtEgral}
\begin{split}
\frac{d}{dt}E_H&=\frac{A_t^{1/2}}{(16\pi)^{3/2}}\left[8\pi- \frac{1}{2}\int_{S_t}(H^2-q^2) dS \right] \\
&-  \frac{A_t^{1/2}}{(16\pi)^{3/2}} \int_{S_t}\left(2H\frac{dH}{dt} -2q\frac{dq}{dt} + (H^2-q^2)\right) dS \\
\end{split}
\ee

First we calculate the derivate of $H$ along the flow, we refer the reader to \cite{Huisken1997} and \cite{Malec:2002ki} for details, proofs and further references. From \eqref{evolH} we have:
\be
2H\frac{dH}{dt}=-2H\Delta(H^{-1})-2h_{ij}h^{ij}+ 2\bar R_{ij}\nu^i\nu^j).
\ee
then we use the Gauss equation
\be
2\bar R_{ij}\nu^i\nu^j=\bar R+H^2-h_{ij}h^{ij}-2\kappa
\ee
where $\kappa$ is the Gauss curvature, and we obtain:
\be
2H\frac{dH}{dt}=-2H\Delta(H^{-1})-h_{ij}h^{ij} -\bar R -H^2 +2\kappa.
\ee
Now let $t_{ij}$ be the trace free part of $h_{ij}$:
\be\label{t}
t_{ij}=h_{ij}-\frac{H}{2}g_{ij}
\ee
then
\be
h_{ij}h^{ij}=t_{ij}t^{ij} + \frac{H^2}{2}
\ee
hence using this and equation \eqref{barR} we obtain:
\be\label{dtH}
\begin{split}
2H\frac{dH}{dt}=&  +2\kappa + 2zq  + \frac{q^2}{2} -2H\Delta(H^{-1}) - \frac{3}{2}H^2\\
& -16 \pi \mu  -2s_i s^i - \chi_{ij} \chi^{ij}  - t_{ij} t^{ij} 
\end{split}
\ee

For the derivate of $q$ along the flow we first calculate the covariant derivate in the direction of $\nu^i$, and then using the vector constraint \eqref{const2} we obtain:
\be\label{dtqgral}
\begin{split}
H\frac{dq}{dt} &= \nu^i \bar \nabla_i \left(\mbox{tr}K - z\right)=\nu^i \bar \nabla_j K^j_i + 8\pi \nu^i j_i - \nu^i \bar \nabla_i z\\
               &=  \bar \nabla_j \left( \nu^i  K^j_i \right) -  K^j_i \bar \nabla_j  \nu^i    + 8\pi \nu^i j_i - \nu^i \bar \nabla_i z\\
               &=  \bar \nabla_j \left( z\nu^j + s^j \right) -  K^j_i \bar \nabla_j \nu^i  + 8\pi \nu^i j_i - \nu^i \bar \nabla_i z\\
               &=  z\bar \nabla_j \nu^j + \bar \nabla_j s^j -  K^j_i \bar \nabla_j \nu^i   + 8\pi \nu^i j_i \\
               &=  zH + \bar \nabla_i s^i -  K^j_i \bar \nabla_j \nu^i   + 8\pi \nu^i j_i \\
\end{split}
\ee
One can write the term $\bar \nabla_i s^i$ in the form:
\be\label{dsi}
\begin{split}
 \bar \nabla_i s^i&= \bar g^{ij} \bar \nabla_i s_j= g^{ij} \bar \nabla_i s_j + \nu^i \nu^j \bar\nabla_i s_j\\
                  &=g^{ij} \bar \nabla_i s_j -  s_j \nu^i  \bar \nabla_i \nu^j \\
                  &=g^{ij} \bar \nabla_i s_j +  H s^j  \bar \nabla_j \frac{1}{H} \\
\end{split}
\ee
where in the last step we use the fact that $\nu^i  \bar \nabla_i \nu_j= - H g^i_j \bar \nabla_i \frac{1}{H}$, see for example \cite{Huiskenevol}. And the term $K^i_j \bar \nabla_i \nu^i$ is:
\be\label{Kh}
\begin{split}
 K^{ij} \bar \nabla_i \nu_i&= z \nu^i \nu^j \bar \nabla_i \nu_j + \left( s^i\nu^j + \nu^i s^j\right)\bar \nabla_i \nu_j + K^{ij}h_{ij} \\
                           &= z \nu^i \nu^j \bar \nabla_i \nu_j + \left( s^i\nu^j + \nu^i s^j\right)\bar \nabla_i \nu_j + \chi^{ij} t_{ij} + \frac{Hq}{2} \\
                           &= - H s^j \bar \nabla_j \frac{1}{H} + \chi^{ij} t_{ij} + \frac{Hq}{2} \\
\end{split}
\ee
where in the last step we use  $s^i\nu^j  \bar \nabla_i \nu_j= \frac{1}{2} s^i \bar \nabla_i (\nu^j \nu_j)=0 $ and equation \eqref{eqIMCF}.

Then using \eqref{dsi} and \eqref{Kh} in \eqref{dtqgral}  we obtain:
\be \label{dtq}
H\frac{dq}{dt}=  8\pi \nu^i j_i+ zH - \frac{Hq}{2} + g^{ij} \bar \nabla_i s_j + 2 H s^j  \bar \nabla_j \frac{1}{H} -  \chi^{ij} t_{ij}  \\
\ee

Now using  this, the integrand on the second term of \eqref{dtEgral} is:
\be\label{dtinEH}
\begin{split}
2H\frac{dH}{dt}& -2q\frac{dq}{dt} + (H^2-q^2) = \\
               =&  2\kappa + 2zq  - \frac{q^2}{2} -2H\Delta(H^{-1}) - \frac{3}{2}H^2 -16 \pi \mu  -2s_i s^i - \chi_{ij} \chi^{ij}  - t_{ij} t^{ij} \\
               &-\frac{q}{H}16\pi \nu^i j_i - 2zq + q^2 - 2 \frac{q}{H} g^{ij} \bar \nabla_i s_j - 4 q s^j  \bar \nabla_j \frac{1}{H} +  2\frac{q}{H}\chi^{ij} t_{ij} + (H^2-q^2) \\
              =&  2\kappa  - \frac{1}{2}(H^2-q^2) -2H\Delta(H^{-1})  -16 \pi \mu  -2s_i s^i - \chi_{ij} \chi^{ij}  - t_{ij} t^{ij} \\
               &-\frac{q}{H}16\pi \nu^i j_i - 2 \frac{q}{H} g^{ij} \bar \nabla_i s_j - 4 q s^j  \bar \nabla_j \frac{1}{H} +  2\frac{q}{H}\chi^{ij} t_{ij} \\
              =& 2\kappa  - \frac{1}{2}(H^2-q^2)  -16 \pi ( \mu +\frac{q}{H}\nu^i j_i) \\
              &- \left(\chi_{ij} \chi^{ij} -  2\frac{q}{H}\chi^{ij} t_{ij}  + t_{ij}   t^{ij}\right) \\ 
              &-2s_i s^i - 4 q s^j  \bar \nabla_j \frac{1}{H}  -2H\Delta(H^{-1}) - 2 \frac{q}{H} g^{ij} \bar \nabla_i s_j\\
\end{split}
\ee

Next, we incorporate the previous expression on the derivate of Hawking energy \eqref{dtEgral}, we use the Gauss-Bonnet theorem and integrate by parts the Laplace operator:
\begin{multline}
\frac{d}{dt}E_H=\frac{A_t^{1/2}}{(16\pi)^{3/2}}\left[8\pi-4\pi\chi(S_t)  +  \int_{S_t} \left(16 \pi ( \mu +\frac{q}{H}\nu^i j_i)  \right)dS 
\right.\\
\left.       + \int_{S_t} 2 \frac{q}{H} g^{ij} \bar \nabla_i s_j dS       +\int_{S_t} \left(\chi_{ij} \chi^{ij} -  2\frac{q}{H}\chi^{ij} t_{ij}  + t_{ij}   t^{ij}\right)dS \right.\\
\left.              +\int_{S_t} 2\left(s_i s^i - 2 \frac{q}{H} s^j \frac{ \bar \nabla_j H}{H} + \frac{ g^{ij} \bar \nabla_i H \bar \nabla_j H}{H^2} \right)dS \right]
\end{multline}

Finally because $S_t$ is assumed to have spherical topology we obtain \eqref{dtEgralf}, and thus assuming either of the conditions, \ref{a} or \ref{b}, presented in Theorem \ref{theo1} we have that Hawking energy is monotonic along the IMCF.

\end{document}